\documentclass[11pt,a4paper]{article}

\usepackage[T1]{fontenc}
\usepackage[utf8]{inputenc}
\usepackage[english]{babel}
\usepackage{csquotes}
\usepackage{lmodern}
\usepackage{microtype}

\usepackage{authblk}

\usepackage{amsmath}
\usepackage{amssymb}
\usepackage{amsfonts}
\usepackage{mathtools}
\usepackage{bm}
\usepackage{braket}
\usepackage{amsthm}

\usepackage{graphicx}
\usepackage{booktabs}
\usepackage{xcolor}
\usepackage[hidelinks]{hyperref}
\usepackage[
backend=biber,
style=numeric,
sorting=none,
maxbibnames=99
]{biblatex}

\newtheorem{proposition}{Proposition}
\newtheorem{remark}{Remark}

\newcommand{\ii}{\mathrm{i}}
\newcommand{\ee}{\mathrm{e}}

\newcommand{\Hilb}{\mathcal H}
\newcommand{\Aalg}{\mathcal A}

\newcommand{\Gset}{\mathcal G}

\newcommand{\Id}{\mathbb I}
\newcommand{\Tr}{\operatorname{Tr}}
\newcommand{\Span}{\operatorname{span}}

\newcommand{\Nop}{\hat N}

\newcommand{\Diss}{\mathcal D}

\newcommand{\R}{\mathbb R}

\title{
Environment-aware transverse transport geometry\\
for non-Gaussian oscillator states
}

\author[]{Arnaud Coatanhay and Angélique Drémeau}

\affil[]{%
Lab-STICC, UMR CNRS 6285, ENSTA, Institut Polytechnique de Paris,\\
2 rue Fran\c{c}ois Verny, 29806 Brest Cedex 9, France
}

\affil[]{%
\href{mailto:arnaud.coatanhay@ensta.fr}{arnaud.coatanhay@ensta.fr}
\quad
\href{mailto:angelique.dremeau@ensta.fr}{angelique.dremeau@ensta.fr}
}

\date{\today}

\begin{document}

\maketitle

\begin{abstract}
Non-Gaussian oscillator states are highly sensitive to weak perturbations, but a reversible phase rotation and an irreversible environmental change should not receive the same diagnostic weight. We introduce a finite-dimensional, direction-dependent local transport cost on a truncated Fock space and quotient it by Hamiltonian tangent directions. The resulting transverse square root is a seminorm, not a global distance between states. A weighted operator frame $(\Gset,w)$ specifies how fluxes are represented; it must be distinguished from a unique Gorini--Kossakowski--Sudarshan--Lindblad (GKSL) representation. The finite-frame covariance $V_j\mapsto c_jV_j$, $w_j\mapsto |c_j|^2w_j$ leaves the cost invariant. We also separate algebraic accessibility from representation efficiency: already $\Gset_1=\{a_M,a_M^\dagger\}$ spans the full trace-zero sector through the divergence map, whereas adding multiphoton or diagonal directions can make selected physical tangents dramatically less expensive and more directly adapted. Numerical diagnostics for compass-like states show that the Hamiltonian phase-rotation tangent is removed, dephasing remains costly until diagonal directions are added, and two- and three-photon Lindblad tangents become inexpensive in a multiphoton frame. Thermal weights are justified at the operator-modular level but, because the mobility is arithmetic, do not define a complete KMS geometry. Finally, an observable-side calculation is presented only as a cotangent diagnostic on the diagonal Fock graph; the full transverse dual is finite only for observables commuting with the reference state. The framework is therefore an environment-aware local diagnostic, with explicit cutoff, regularization, image-residual and normalization conventions.
\end{abstract}

\section{Introduction}
\label{sec:introduction}

Non-Gaussian states of a harmonic oscillator are central resources in continuous-variable quantum optics, quantum information and quantum sensing. They include Fock states, Schrödinger-cat states, compass states and states produced by nonlinear or multiphoton processes. Their usefulness comes from features that are absent from Gaussian descriptions: interference fringes, parity structure, oscillatory Fock coherences and Wigner negativity \cite{weedbrook_pirandola_garcia_patron_cerf_ralph_shapiro_lloyd_2012,walschaers_2021}. These features make such states sensitive probes, but they also make the meaning of a state variation less obvious.

The difficulty is already visible in a simple example. A phase rotation of a cat or compass state can strongly change its Wigner function. The lobes and interference fringes move in phase space, so a usual norm may see a sizable change. Yet the motion is reversible. By contrast, dephasing may wash out the same fringes and irreversibly suppress coherences. Loss or gain changes the Fock population structure. For an open oscillator system, these variations should not be assigned the same physical status.

This is the main question of the paper. Given a state $\rho$, an infinitesimal variation $X=\dot\rho$, and a specified environment, can we remove the reversible Hamiltonian component and quantify how efficiently the remainder is represented by an environment-adapted operator frame? A one-photon frame, a multiphoton frame and a dephasing-enriched frame should not assign the same local cost.

The Wigner representation is not the best object on which to define this geometry directly. It is an excellent diagnostic: it displays interference, negativity and phase-space structure. However, the Wigner function is a quasi-probability distribution and may take negative values \cite{hudson_1974,kenfack_zyczkowski_2004,albarelli_genoni_paris_ferraro_2018}. We therefore keep Wigner functions for visualization, but build the geometry on the density matrix and on the finite operator algebra of a truncated oscillator.

We work in the finite Fock space
\[
\Hilb_M=\Span\{\ket{0},\ldots,\ket{M}\},
\qquad
\Aalg_M=\mathcal B(\Hilb_M),
\]
and choose a finite family of physical directions
\[
\Gset=\{V_j\}_j.
\]
Typical examples are one- and multiphoton lowering operators
\[
L_k=\Pi_Ma^k\Pi_M,
\]
their adjoints, and diagonal directions associated with dephasing. The set \(\Gset\) is part of the physical model. It says which elementary changes are available to the environment or to an engineered control mechanism.

From \(\Gset\) we define a noncommutative divergence, a continuity equation and a regularized quadratic flux cost. The raw local cost
\[
\mathfrak g_{\rho,\varepsilon}^{\Gset,w}(X,X)
\]
measures how difficult it is to represent the tangent vector \(X\) by fluxes along the directions in \(\Gset\). This is analogous in spirit to the dynamical formulation of optimal transport \cite{benamou_brenier_2000,villani_2003}, and to noncommutative transport metrics associated with quantum Markov semigroups \cite{carlen_maas_2017,wirth_2018}. The emphasis here is different: we do not start from a given semigroup and construct a metric for its entropy gradient flow. We start from selected physical directions and ask how efficiently they represent selected tangents.

The key step is to remove reversible Hamiltonian motion. The unitary tangent space at \(\rho\) is
\[
T_\rho^{\mathrm{unit}}
=
\{-\ii[H,\rho]:H=H^\dagger\}.
\]
We define the transverse cost by
\[
\mathfrak g_{\rho,\varepsilon}^{\perp,\Gset,w}(X,X)
=
\inf_{H=H^\dagger}
\mathfrak g_{\rho,\varepsilon}^{\Gset,w}
\left(
X+\ii[H,\rho],
X+\ii[H,\rho]
\right).
\]
Thus the geometry first subtracts the best reversible motion and then measures the remaining dissipative component. In particular, a phase rotation has zero transverse cost, while dephasing, loss and gain generally remain visible.

This construction should be read as a toolbox rather than as a universal metric. An experimentalist does not have to assume that all state-space directions are equally meaningful. The direction family \(\Gset\) can be chosen from the mechanisms that are actually present, controlled or suspected in a given optical setup. A passive lossy device may suggest \(\Gset_1=\{L_1,L_1^\dagger\}\). A nonlinear or engineered reservoir may require \(L_2\), \(L_3\) or their adjoints. If phase diffusion is relevant, diagonal or number-operator directions can be added. Once \(\rho\), \(X\) and \(\Gset\) have been specified, the transverse cost answers a concrete diagnostic question: after subtracting the best reversible Hamiltonian motion, can the observed variation be produced efficiently by the experimentally available directions?

For example, suppose that a cat-like or compass-like state is prepared and that tomography, or a reduced reconstruction, gives a small variation after propagation through an optical device. A standard norm may indicate that the state has changed, but it does not directly say whether the change is a harmless phase rotation, ordinary attenuation, multiphoton loss, or dephasing. The proposed diagnostic compares the same variation against several candidate direction families. A near-zero transverse cost indicates an essentially reversible motion. A low but nonzero cost indicates compatibility with the chosen dissipative model. A very large cost suggests that an important physical direction is missing. In this sense the geometry can be used as a model-selection diagnostic for the experimental environment.

This is useful in three concrete ways. First, it prevents a unitary distortion of a non-Gaussian pattern from being confused with environmental degradation. Second, it tests whether a proposed perturbation is compatible with a specified set of physical channels, for example whether two- or three-photon losses are available. Third, through a restricted cotangent diagnostic on the diagonal Fock graph, it indicates how the chosen transition structure weights selected diagonal observables. This connects the present work with our previous measurement-aware viewpoint, where useful distinguishability depended on the measurement chain \cite{our_measurement_aware_2026}. Here the context is instead the environment.

The resulting object is not intended as a universal quantum Wasserstein distance. Quantum Wasserstein distances and noncommutative transport metrics address broader questions, including entropy gradient flows, channel-based transport and Lipschitz observable dualities \cite{carlen_maas_2017,wirth_2018,depalma_marvian_trevisan_lloyd_2021,depalma_trevisan_2021}. The present construction is narrower: it is a local, regularized and environment-aware diagnostic for infinitesimal variations of non-Gaussian oscillator states.

The contribution of the paper is therefore practical rather than axiomatic. We define the local cost, quotient it by unitary directions, and test it numerically on representative oscillator states. The diagnostics show that phase rotations are removed, dephasing requires diagonal directions, multiphoton loss tangents become inexpensive when the corresponding transitions are included explicitly, thermal or modular weights polarize gain and loss, and the same direction structure modifies the sensitivity of Fock observables.

The framework is deliberately finite-dimensional. The cutoff \(M\), the regularization \(\varepsilon\), the direction family \(\Gset\), the weights \(w\) and the reference state are part of the modeling. This model-dependence is not a weakness of the construction: it is what allows the diagnostic to reflect the available experimental toolbox. The goal is not to solve an infinite-dimensional optimal transport problem for oscillator states. It is to obtain a computable diagnostic of dissipative representation cost adapted to a specified physical environment.

The paper is organized as follows. Section~\ref{sec:physical-problem} formulates the physical distinction between reversible and dissipative variations. Section~\ref{sec:fock-model} introduces the finite Fock-space setting and the multiphoton directions. Section~\ref{sec:local-cost} defines the local noncommutative transport cost. Section~\ref{sec:transverse-quotient} introduces the transverse quotient by unitary orbits. Section~\ref{sec:numerics} presents the numerical diagnostics. Section~\ref{sec:relations} discusses the relation with existing quantum Wasserstein viewpoints. Section~\ref{sec:conclusion} concludes.

\section{Physical problem: dissipative variations of non-Gaussian oscillator states}
\label{sec:physical-problem}

We consider a single harmonic oscillator prepared in a non-Gaussian state and subject to an environment that may induce loss, gain, dephasing or multiphoton transitions. The state is described by a density matrix \(\rho\). We are interested in infinitesimal variations
\[
X=\dot\rho
\]
around \(\rho\).

The central point is that not all tangent vectors have the same physical meaning. A variation may be reversible and generated by a Hamiltonian, or it may correspond to an irreversible environmental process. A geometry intended to diagnose dissipation should not treat these two cases in the same way.

\subsection{Reversible phase rotations}

The simplest reversible variation is a phase rotation generated by the number operator:
\[
X_{\mathrm{rot}}
=
-\ii[\Nop,\rho].
\]
This tangent vector changes the phase of the oscillator state but preserves the spectrum of \(\rho\). It lies on the unitary orbit of the state:
\[
\rho
\longmapsto
\ee^{-\ii t\Nop}\rho\,\ee^{\ii t\Nop}.
\]
Such a variation may modify the Wigner function by rotating its interference pattern in phase space, but it does not by itself represent dissipation or loss of information to an environment.

For this reason, a purely state-space norm of \(X_{\mathrm{rot}}\) is not the relevant quantity if the goal is to quantify dissipative degradation. The variation can be large as a tangent vector and still be physically reversible.

\subsection{Dephasing and loss of coherence}

By contrast, phase diffusion or dephasing is represented here by the fixed convention
\[
X_{\mathrm{deph}}
=
-[\Nop,[\Nop,\rho]]
=
2\Diss[\Nop](\rho),
\qquad
\Diss[V](\rho)=V\rho V^\dagger-\frac12\{V^\dagger V,\rho\}.
\]
The factor of two is retained throughout the paper and in all numerical figures. In the Fock basis, the corresponding semigroup suppresses off-diagonal coherences according to
\[
\rho_{mn}(t)=\ee^{-t(m-n)^2}\rho_{mn}(0)
\]
for the normalization used above. This is not a unitary motion: it changes the coherence structure and generally degrades non-Gaussian interference features.

For cat or compass-like states, the distinction is crucial. A phase rotation merely reorients the interference pattern, whereas dephasing washes it out. Both effects are visible in phase space, but only the second is irreversible.

\subsection{Multiphoton loss and gain}

More general environmental mechanisms may involve one-photon or multiphoton transitions. For a lowering operator \(L_k\), the associated Lindblad dissipator is
\[
\Diss[L_k](\rho)
=
L_k\rho L_k^\dagger
-
\frac12
\left\{
L_k^\dagger L_k,
\rho
\right\}.
\]
The corresponding infinitesimal variation is
\[
X_{\mathrm{loss},k}
=
\Diss[L_k](\rho).
\]
Similarly, the gain direction associated with \(L_k^\dagger\) is
\[
X_{\mathrm{gain},k}
=
\Diss[L_k^\dagger](\rho).
\]

These variations are not simply geometric displacements of the state. They encode physically distinct environmental processes. One-photon loss, two-photon loss and three-photon loss may affect the same state in very different ways, especially when the state has parity structure or multiphoton coherence.

This motivates a direction-dependent geometry. The cost assigned to a variation should depend on whether the chosen environmental frame represents the corresponding transition directly and economically.

\subsection{Toolbox interpretation}

The diagnostic proposed below answers a local representation question:
\[
\begin{aligned}
(\rho,X,\Gset,w)
&\longmapsto
\text{minimum weighted flux cost}\\
&\hspace{2.8cm}\text{of the non-Hamiltonian part of }X.
\end{aligned}
\]
Two logically different notions must be separated. \emph{Algebraic accessibility} asks whether $X$ belongs to the image of the divergence. \emph{Representation efficiency} asks how large the minimizing quadratic cost is in the chosen weighted frame. In the finite truncation used here, $\Gset_1=\{a_M,a_M^\dagger\}$ already has the full trace-zero sector as its divergence image; adding $a_M^2$, $a_M^3$ or diagonal operators does not create algebraic accessibility for the examples below. It changes the conditioning and physical adaptation of the representation, sometimes by many orders of magnitude.

The set $\Gset$ and the positive weights $w$ are chosen by the experimentalist or modeler. They may encode known channels, engineered controls, or competing hypotheses about the dominant perturbation. Zero transverse cost means that the tangent is Hamiltonian within the quotient. A small nonzero cost means that its non-Hamiltonian component is represented economically by the selected frame. A large cost indicates an inefficient or poorly adapted frame, not necessarily an algebraically impossible tangent.

The construction is therefore not a distance between arbitrary density matrices. It is a state-dependent local cost and, after quotienting, a transverse seminorm for comparing reversible reorientation with dissipative coherence or population changes.

\section{Finite Fock-space model and multiphoton directions}
\label{sec:fock-model}

We work in a finite Fock-space truncation. Let
\[
\Hilb_M
=
\Span\{\ket{0},\ket{1},\ldots,\ket{M}\},
\qquad
d=M+1,
\]
and let
\[
\Aalg_M=\mathcal B(\Hilb_M)
\]
be the finite-dimensional algebra of operators acting on \(\Hilb_M\). Density matrices are positive elements of \(\Aalg_M\) with unit trace.

The projection onto \(\Hilb_M\) is denoted by \(\Pi_M\). The truncated annihilation, creation and number operators are
\[
a_M=\Pi_M a\Pi_M,
\qquad
a_M^\dagger=\Pi_M a^\dagger\Pi_M,
\qquad
\Nop_M=\Pi_M \Nop \Pi_M.
\]
When no confusion is possible, we omit the subscript \(M\).

The finite cutoff has two roles. First, it makes all objects finite matrices, so that the transport problem becomes a finite-dimensional constrained quadratic optimization problem. Second, it reflects the fact that numerical and experimental descriptions of non-Gaussian oscillator states are always effectively energy-limited.

\subsection{Multiphoton transition directions}

For \(k\geq 1\), we define the truncated \(k\)-photon lowering operator
\[
L_k
=
\Pi_M a^k\Pi_M,
\]
and its adjoint
\[
L_k^\dagger
=
\Pi_M (a^\dagger)^k\Pi_M.
\]
The action of \(L_k\) connects Fock levels separated by \(k\):
\[
L_k\ket{n}
=
\sqrt{n(n-1)\cdots(n-k+1)}
\ket{n-k},
\]
whenever \(n\geq k\), and gives zero otherwise.

These directions represent elementary multiphoton transitions. The case \(k=1\) corresponds to ordinary one-photon loss or gain. Higher values of \(k\) model nonlinear or engineered processes such as two-photon or three-photon transitions.

A basic family of admissible directions is therefore
\[
\Gset_K
=
\{L_k,L_k^\dagger:1\leq k\leq K\}.
\]
For example,
\[
\Gset_1
=
\{L_1,L_1^\dagger\}
\]
contains only one-photon transitions, while
\[
\Gset_3
=
\{L_1,L_1^\dagger,L_2,L_2^\dagger,L_3,L_3^\dagger\}
\]
also allows two-photon and three-photon transitions.

\subsection{Algebraic accessibility versus representation efficiency}

The finite truncation makes an important distinction explicit.

\begin{proposition}[Image of the one-photon frame]
For $M\geq1$ and
\[
\Gset_1=\{a_M,a_M^\dagger\},
\]
the divergence image is the full trace-zero matrix sector:
\[
\operatorname{Im}(\operatorname{div}_{\Gset_1})
=
\{X\in\Aalg_M:\Tr X=0\}.
\]
Consequently every Hermitian traceless tangent vector is algebraically representable already with $\Gset_1$.
\end{proposition}

\begin{proof}
Every divergence term is a commutator and therefore has zero trace. Conversely, the Hilbert--Schmidt orthogonal complement of the image consists of matrices commuting with both $a_M$ and $a_M^\dagger$. Their common commutant is scalar: commuting with $a_M^\dagger a_M=\Nop_M$ makes a matrix diagonal in the nondegenerate Fock basis, and commuting with $a_M$ forces all diagonal entries to coincide. Hence the image has codimension one and equals the trace-zero sector.
\end{proof}

This proposition does not make all directions equally economical. The mobility and weighted frame determine the minimum flux norm. Higher-order or diagonal operators can therefore make a tangent \emph{inexpensive} or \emph{directly adapted} even though it was already algebraically representable with $\Gset_1$.

\subsection{Dephasing and diagonal directions}

Multiphoton loss and gain are not sufficient to represent all relevant dissipative effects. In particular, dephasing acts mainly on coherences in the Fock basis. A natural dephasing generator is the number operator \(\Nop_M\), leading to the Lindblad-type variation
\[
\Diss[\Nop_M](\rho)
=
\Nop_M\rho\Nop_M
-
\frac12
\{\Nop_M^2,\rho\}.
\]
Equivalently,
\[
2\Diss[\Nop_M](\rho)
=
-[\Nop_M,[\Nop_M,\rho]]
=
X_{\mathrm{deph}}.
\]
This equality fixes the normalization used in the numerical section.

It may also be useful to add diagonal directions. A simple choice is a family of projectors or finite differences in the Fock basis, for instance
\[
P_n=\ket{n}\bra{n},
\qquad
0\leq n\leq M,
\]
or suitable centered combinations of such projectors. These directions improve the representation of population changes and dephasing-like variations.

We shall use the notation
\[
\Gset_K^{(N)}
=
\Gset_K\cup\{\Nop_M\},
\]
and
\[
\Gset_K^{\mathrm{diag}}
=
\Gset_K\cup\{\text{diagonal directions}\},
\]
when diagonal or dephasing directions are included.

\subsection{Physical meaning of the weighted operator frame}

The pair $(\Gset,w)$ is part of the model. It is a weighted operator frame for the continuity equation, not a unique GKSL representation of an environment. Different jump-operator parameterizations can describe the same GKSL generator, and the transport cost must not depend on a purely scalar reparameterization.

For nonzero complex numbers $c_j$, define
\[
V_j'=c_jV_j,
\qquad
w_j'=|c_j|^2w_j.
\]
With the flux change $J_j'=J_j/\overline{c_j}$, both the divergence and the quadratic flux cost are unchanged. Therefore
\[
\mathfrak g_{\rho,\varepsilon}^{\Gset',w'}(X,X)
=
\mathfrak g_{\rho,\varepsilon}^{\Gset,w}(X,X).
\]
This covariance is consistent with GKSL rescaling: the term $\gamma_j\Diss[V_j]$ is unchanged under $V_j\mapsto c_jV_j$ and $\gamma_j\mapsto\gamma_j/|c_j|^2$; if $w_j\propto1/\gamma_j$, then $w_j\mapsto|c_j|^2w_j$.

Beyond such frame covariance, changing the operator content genuinely changes representation efficiency. Ordinary attenuation suggests a one-photon frame. Engineered nonlinear reservoirs suggest adding higher-order transitions. Phase noise suggests adding number-operator or diagonal directions. Since $\Gset_1$ already spans the trace-zero sector algebraically, the purpose of these additions is not to turn an impossible tangent into a possible one, but to provide a more economical and physically interpretable representation.

\subsection{Test states}

The numerical part of the paper will use a small set of representative truncated oscillator states. These include coherent states, Fock states and non-Gaussian superpositions such as cat or compass-like states.

Coherent states provide a quasi-classical reference. Fock states provide states with sharply defined occupation number and strong non-Gaussianity. Cat and compass-like states provide phase-space interference structures and are particularly sensitive to dephasing and parity-related effects.

The purpose is not to classify all non-Gaussian states, but to test whether the transverse transport cost reacts differently to reversible rotations, dephasing, loss, gain and multiphoton transitions.

\section{Local noncommutative transport cost}
\label{sec:local-cost}

We now define the local cost associated with a given family of physical directions. The construction is finite-dimensional and local around a state \(\rho\). It is designed to answer the following question: how costly is it to represent a tangent vector \(X\) using fluxes along the admissible directions in \(\Gset\)?

\subsection{Tangent vectors}

The state space in the finite truncation is
\[
\mathcal D_M
=
\left\{
\rho\in\Aalg_M:
\rho=\rho^\dagger,\ 
\rho\geq 0,\ 
\Tr(\rho)=1
\right\}.
\]
At an interior point, tangent vectors are Hermitian traceless matrices:
\[
T_\rho\mathcal D_M
=
\left\{
X\in\Aalg_M:
X=X^\dagger,\ 
\Tr(X)=0
\right\}.
\]
The construction below also applies to boundary states after regularization.

\subsection{Regularized mobility}

For nearly pure states, the density matrix \(\rho\) may have zero eigenvalues. We therefore introduce a regularized state
\[
\rho_\varepsilon
=
(1-\varepsilon)\rho
+
\varepsilon\frac{\Id}{d},
\qquad
0<\varepsilon<1,
\]
where \(d=M+1\). This ensures that \(\rho_\varepsilon\) is strictly positive.

We use the symmetric noncommutative mobility
\[
\mathcal K_{\rho_\varepsilon}(Y)
=
\frac12
\left(
\rho_\varepsilon Y
+
Y\rho_\varepsilon
\right).
\]
Since \(\rho_\varepsilon>0\), the map \(\mathcal K_{\rho_\varepsilon}\) is invertible on \(\Aalg_M\). Other choices, such as logarithmic means, can be used, but the symmetric mobility is sufficient for the present article and leads to a simple quadratic optimization problem.

The Hilbert--Schmidt pairing is
\[
\langle A,B\rangle_{\mathrm{HS}}
=
\Tr(A^\dagger B).
\]
In quadratic costs we use the real part when necessary.

\subsection{Fluxes and divergence}

Let
\[
\Gset=\{V_j\}_{j=1}^r
\]
be a finite family of admissible directions. A flux is a collection of matrices
\[
J=(J_1,\ldots,J_r),
\qquad
J_j\in\Aalg_M.
\]
The divergence associated with \(\Gset\) is defined by
\[
\operatorname{div}_{\Gset}(J)
=
\sum_{j=1}^r
\left(
J_jV_j^\dagger
-
V_j^\dagger J_j
\right).
\]
The local continuity equation is
\[
X+\operatorname{div}_{\Gset}(J)=0.
\]
Thus a tangent vector \(X\) is represented by fluxes along the physical directions \(V_j\).

This is the noncommutative analogue of writing a velocity field as the divergence of a current. The important difference is that the available currents are not arbitrary spatial vector fields; they are operator-valued fluxes tied to the chosen transition directions.

\subsection{Weighted flux cost}

Let
\[
w=(w_1,\ldots,w_r),
\qquad
w_j>0,
\]
be positive resistance weights. The cost of a flux $J$ at the state $\rho$ is
\[
\mathcal C_{\rho_\varepsilon,w}(J)
=
\sum_{j=1}^r
w_j\operatorname{Re}\left\langle
J_j,\mathcal K_{\rho_\varepsilon}^{-1}(J_j)
\right\rangle_{\mathrm{HS}}.
\]
Because $\mathcal K_{\rho_\varepsilon}$ is positive, each term is nonnegative. The model datum is the weighted frame $(\Gset,w)$ rather than the bare list $\Gset$. The simultaneous rescaling $V_j\mapsto c_jV_j$, $w_j\mapsto|c_j|^2w_j$ is a change of frame coordinates and leaves the induced cost invariant.

\subsection{Local cost of a tangent vector}

The local quadratic cost of a tangent vector $X$ is
\[
\mathfrak g_{\rho,\varepsilon}^{\Gset,w}(X,X)
=
\inf_{J:\,X+\operatorname{div}_{\Gset}(J)=0}
\mathcal C_{\rho_\varepsilon,w}(J),
\]
with value $+\infty$ outside the divergence image. For all frames used numerically here, $\Gset_1\subseteq\Gset$, so every Hermitian traceless tangent belongs to that image and the cost is finite up to numerical tolerance.

The construction is homogeneous of degree two:
\[
\mathfrak g_{\rho,\varepsilon}^{\Gset,w}(\lambda X,\lambda X)
=\lambda^2\mathfrak g_{\rho,\varepsilon}^{\Gset,w}(X,X),
\qquad \lambda\in\R.
\]
The same property holds after the transverse quotient. Hence absolute numerical values depend on the physical normalization and units of $X$. Unless explicitly stated otherwise, the figures use the raw tangents written in the text, with no Hilbert--Schmidt or unit-velocity normalization. In particular, $X_{\mathrm{deph}}=2\Diss[N](\rho)$ is used exactly as defined.

The square root
\[
\|X\|_{\rho,\varepsilon;\Gset,w}
=
\sqrt{\mathfrak g_{\rho,\varepsilon}^{\Gset,w}(X,X)}
\]
is a local seminorm on the finite-cost tangent sector. It is not an intrinsic global distance between density matrices.

\subsection{Interpretation}

A low cost means that $X$ has an economical flux representation in the chosen weighted frame. A high cost means that the representation is poorly conditioned or physically indirect. For example, a two-photon loss tangent becomes inexpensive when $L_2$ is present as a direct frame element, although it was already algebraically representable with $\Gset_1$. Likewise, a dephasing tangent becomes inexpensive after adding diagonal directions. This dependence is the intended environment-aware feature.

\section{Transverse quotient: separating unitary and dissipative variations}
\label{sec:transverse-quotient}

The local cost introduced in Sec.~\ref{sec:local-cost} measures the representation efficiency of a tangent vector in the chosen weighted frame $(\Gset,w)$. However, it does not yet distinguish between reversible and dissipative variations. This distinction is essential for the physical problem considered here.

A variation generated by a Hamiltonian is reversible. It moves the state along its unitary orbit and does not correspond, by itself, to loss of information to the environment. We therefore introduce a transverse cost that removes the unitary component of a tangent vector before assigning a dissipative cost.

\subsection{Unitary orbits}

Let \(\rho\in\mathcal D_M\). Its unitary orbit is
\[
\mathcal O_\rho
=
\{U\rho U^\dagger:U\in\mathcal U(\Hilb_M)\}.
\]
This orbit contains all states with the same spectrum as \(\rho\). It represents the reversible Hamiltonian motions available at the level of density matrices.

The tangent space to the unitary orbit at \(\rho\) is
\[
T_\rho^{\mathrm{unit}}
=
\{-\ii[H,\rho]:H=H^\dagger\}.
\]
If \(X\in T_\rho^{\mathrm{unit}}\), then \(X\) changes the eigenbasis of \(\rho\), but not its eigenvalues. Such a variation is therefore reversible.

The complement of \(T_\rho^{\mathrm{unit}}\) is not canonical. We do not choose a fixed linear complement. Instead, we define a quotient cost by minimizing over all unitary tangent directions.

\subsection{Definition of the transverse cost}

Let \(X\in T_\rho\mathcal D_M\). We define the transverse cost of \(X\) by
\[
\mathfrak g_{\rho,\varepsilon}^{\perp,\Gset,w}(X,X)
=
\inf_{H=H^\dagger}
\mathfrak g_{\rho,\varepsilon}^{\Gset,w}
\left(
X+\ii[H,\rho],
X+\ii[H,\rho]
\right).
\]
Equivalently, \(X\) and \(X'\) are regarded as equivalent if
\[
X'-X\in T_\rho^{\mathrm{unit}}.
\]
The transverse cost is the minimal local transport cost among all representatives of this equivalence class.

This construction has a direct physical interpretation. The Hamiltonian part of a variation is treated as reversible and is removed. In the present paper the infimum runs over \emph{all} Hermitian Hamiltonians, so the quotient is by the full unitary orbit. A future model may instead choose an experimentally reversible subspace $\mathfrak h_{\mathrm{rev}}\subseteq\{H=H^\dagger\}$ and restrict the infimum to $H\in\mathfrak h_{\mathrm{rev}}$; this would retain coherent motions that are not available as reversible controls.

\subsection{Vanishing on reversible variations}

The transverse cost vanishes on unitary tangent vectors.

\begin{proposition}[Vanishing on unitary directions]
Let \(X=-\ii[H_0,\rho]\) for some Hermitian operator \(H_0\). Then
\[
\mathfrak g_{\rho,\varepsilon}^{\perp,\Gset,w}(X,X)=0.
\]
\end{proposition}

\begin{proof}
Choose \(H=H_0\) in the definition of the transverse cost. Then
\[
X+\ii[H,\rho]
=
-\ii[H_0,\rho]+\ii[H_0,\rho]
=
0.
\]
The zero tangent vector is represented by the zero flux \(J=0\), whose cost is zero. Hence the infimum is zero.
\end{proof}

This is the basic property that distinguishes the transverse cost from the raw local cost. A phase rotation can have nonzero raw cost but zero transverse cost.

\subsection{Dissipative variations remain transverse}

By contrast, a genuinely dissipative tangent vector is not generally tangent to the unitary orbit. For example, the dephasing direction
\[
X_{\mathrm{deph}}
=
-[\Nop,[\Nop,\rho]]
\]
changes the coherence structure of the state and cannot, in general, be written as
\[
-\ii[H,\rho].
\]
Similarly, Lindblad variations such as
\[
X_{\mathrm{loss},k}
=
\Diss[L_k](\rho)
\]
usually change the spectrum of \(\rho\), and therefore have a nonzero transverse component.

The transverse cost is designed to detect precisely this difference:
\[
\text{unitary rotation}
\quad\Rightarrow\quad
\text{zero transverse cost},
\]
whereas
\[
\text{dephasing or loss}
\quad\Rightarrow\quad
\text{nonzero transverse cost}.
\]

\subsection{Transverse seminorm on the quotient}

Define
\[
\|X\|_{\rho,\varepsilon;\Gset,w,\perp}
=
\sqrt{\mathfrak g_{\rho,\varepsilon}^{\perp,\Gset,w}(X,X)}.
\]
This is a seminorm on the full tangent space and a norm on the quotient by the zero directions of the quadratic form, which include $T_\rho^{\mathrm{unit}}$. It is not a distance between finite-separated states. Its quadratic homogeneity makes the normalization of a tangent explicit:
\[
\|\lambda X\|_{\perp}=|\lambda|\,\|X\|_{\perp}.
\]

\subsection{Local transverse representation cost}

The transverse cost combines an optimization over reversible Hamiltonian motion with a weighted flux representation of the remainder. A tangent has low transverse cost if, after subtracting a suitable Hamiltonian component, it is represented economically by $(\Gset,w)$. Because algebraic accessibility is already complete for $\Gset_1$ in the finite truncation, comparisons between $\Gset_1$, $\Gset_3$ and $\Gset_3^{\mathrm{diag}}$ below are comparisons of efficiency and direct adaptation, not binary reachability tests.

\section{Numerical use cases}
\label{sec:numerics}

The numerical section is written as a minimal experimental workflow. We start from representative states, choose several candidate direction families, and ask which perturbations are compatible with each family after unitary motion has been removed. The five figures should be read as successive tests of the toolbox: Fig.~\ref{fig:test-states} defines the states; Fig.~\ref{fig:raw-vs-transverse} tests reversibility removal; Fig.~\ref{fig:gset-dependence} tests environment selection; Fig.~\ref{fig:thermal-weights} tests thermal or modular weighting; and Fig.~\ref{fig:observable-dual} tests the induced observable sensitivity.

All computations are performed in a truncated Fock space. Unless stated otherwise, the numerical tests use \(M=10\), \(d=M+1\), and \(\varepsilon=10^{-5}\). The values reported below are regularized local costs. Their absolute magnitude depends on the cutoff, the mobility, the weights and the direction family. The robust information is the hierarchy between perturbations and its dependence on \(\Gset\).

The practical protocol is the following. Choose a state $\rho$, candidate perturbations $X$, and one or several weighted frames $(\Gset,w)$, then compare transverse costs using the same tangent normalization. The comparison is not meant to identify a unique microscopic model by itself. It tests which frame represents the measured tangent economically and directly.

\subsection{States and perturbations}

We use four representative states: a coherent state, the Fock state \(\ket{3}\), an even cat state and a compass-like state. They are shown in Fig.~\ref{fig:test-states} both through their Wigner functions and through the modulus of their density matrices in the Fock basis. The coherent state gives a quasi-classical reference. The Fock state is non-Gaussian without phase-space interference fringes. The cat and compass-like states display coherence patterns and Wigner structures that are sensitive to dephasing and multiphoton transitions.

For a state \(\rho\), the perturbations used in the diagnostics are
\[
X_{\mathrm{rot}}=-\ii[\Nop,\rho],
\qquad
X_{\mathrm{deph}}=-[\Nop,[\Nop,\rho]],
\]
\[
X_{\mathrm{loss},k}=\Diss[L_k](\rho),
\qquad
X_{\mathrm{gain},k}=\Diss[L_k^\dagger](\rho).
\]
The first perturbation is Hamiltonian. The other perturbations are dissipative in the sense of open quantum dynamics \cite{lindblad_1976,gorini_kossakowski_sudarshan_1976,breuer_petruccione_2002,gardiner_zoller_2004}.

For a short experimental evolution, tomography or a reduced reconstruction provides the finite-difference estimate
\[
X_{\mathrm{exp}}
\simeq
\frac{\rho_{\mathrm{out}}-\rho_{\mathrm{in}}}{\Delta t}.
\]
Before evaluating the cost, statistical reconstruction errors may be removed by Hermitian symmetrization and trace-zero projection. Because the cost is quadratic, omitting the factor $1/\Delta t$ multiplies all reported costs by $\Delta t^2$. The time step and any additional tangent normalization must therefore be stated explicitly.

\begin{figure}[ht]
\centering
\includegraphics[width=0.6\linewidth]{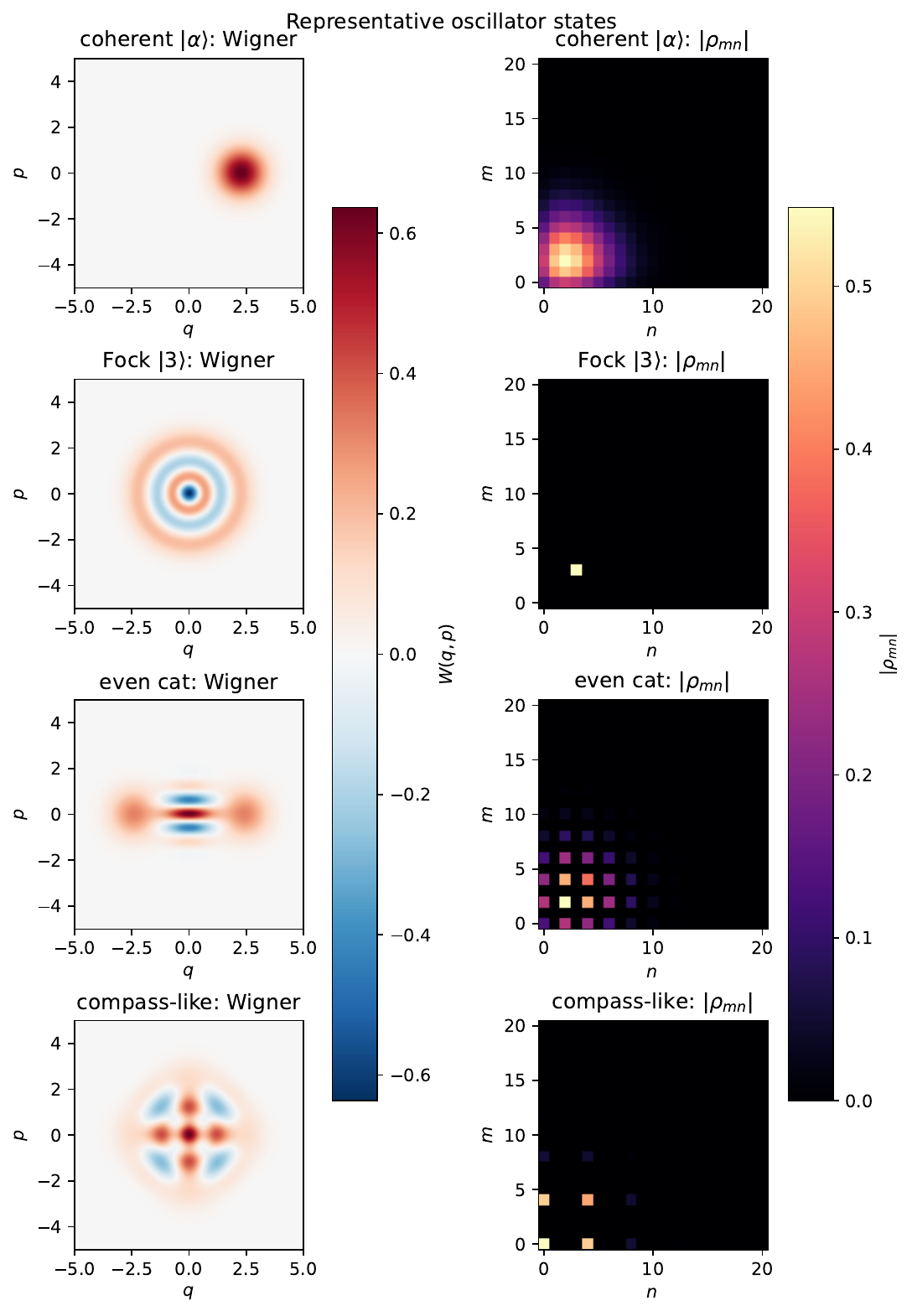}
\caption{
Representative oscillator states used in the numerical diagnostics. Each state is shown through its Wigner function and through the modulus of its density matrix in the Fock basis. The Wigner plots are used only as diagnostics; the transport geometry itself is defined on the density matrix and on the operator directions in Fock space.
}
\label{fig:test-states}
\end{figure}

\subsection{Raw cost versus transverse cost}

The first diagnostic isolates the Hamiltonian quotient. We compare the raw and transverse quadratic costs for the compass-like state in the unit-weight frame
\[
\Gset_3=\{L_1,L_1^\dagger,L_2,L_2^\dagger,L_3,L_3^\dagger\}.
\]
All bars use the raw physical tangents, without normalization. Figure~\ref{fig:raw-vs-transverse} shows that the phase rotation has raw cost $1.13\times10^{-1}$ but transverse cost $7.3\times10^{-27}$. Dephasing remains at $2.25\times10^5$, while the one-, two- and three-photon loss tangents have nonzero transverse costs $1.63$, $4.71$ and $11.2$. The quotient therefore removes the reversible component without normalizing away dissipative amplitudes.

\begin{figure}[t]
\centering
\includegraphics[width=0.6\linewidth]{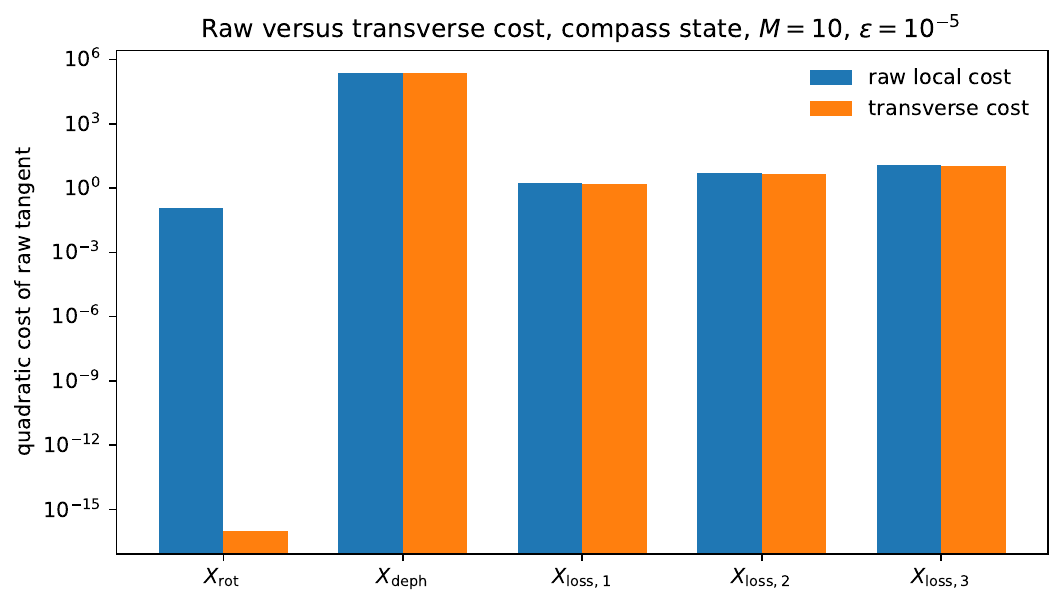}
\caption{
Raw and transverse quadratic costs for raw, unnormalized tangents of a compass-like state, with $M=10$, $\varepsilon=10^{-5}$ and frame $\Gset_3$ at unit weights. The phase rotation is removed by the quotient; dephasing and Lindblad loss tangents remain nonzero.
}
\label{fig:raw-vs-transverse}
\end{figure}

\subsection{Dependence on the weighted frame}

The second diagnostic compares
\[
\Gset_1=\{L_1,L_1^\dagger\},
\qquad
\Gset_3=\{L_1,L_1^\dagger,L_2,L_2^\dagger,L_3,L_3^\dagger\},
\]
and
\[
\Gset_3^{\mathrm{diag}}
=
\Gset_3\cup\{\widetilde N,\widetilde N^2,\widetilde N^3\},
\]
where each diagonal power is centered by removing its identity component. All weights are one and all perturbations retain their raw physical normalization.

Figure~\ref{fig:gset-dependence} must be read as an efficiency comparison. The divergence image has rank $d^2-1=120$ for all three frames. Nevertheless, the transverse cost of $X_{\mathrm{loss},2}$ falls from $3.74\times10^6$ in $\Gset_1$ to $4.71$ in $\Gset_3$, and the cost of $X_{\mathrm{loss},3}$ falls from $6.32\times10^6$ to $11.2$. The corresponding multiphoton operators therefore make these tangents inexpensive and directly adapted; they do not make them algebraically accessible for the first time. Dephasing remains expensive in $\Gset_3$ at $2.25\times10^5$ and becomes inexpensive in $\Gset_3^{\mathrm{diag}}$ at $6.76\times10^{-1}$.

\begin{figure}[t]
\centering
\includegraphics[width=0.7\linewidth]{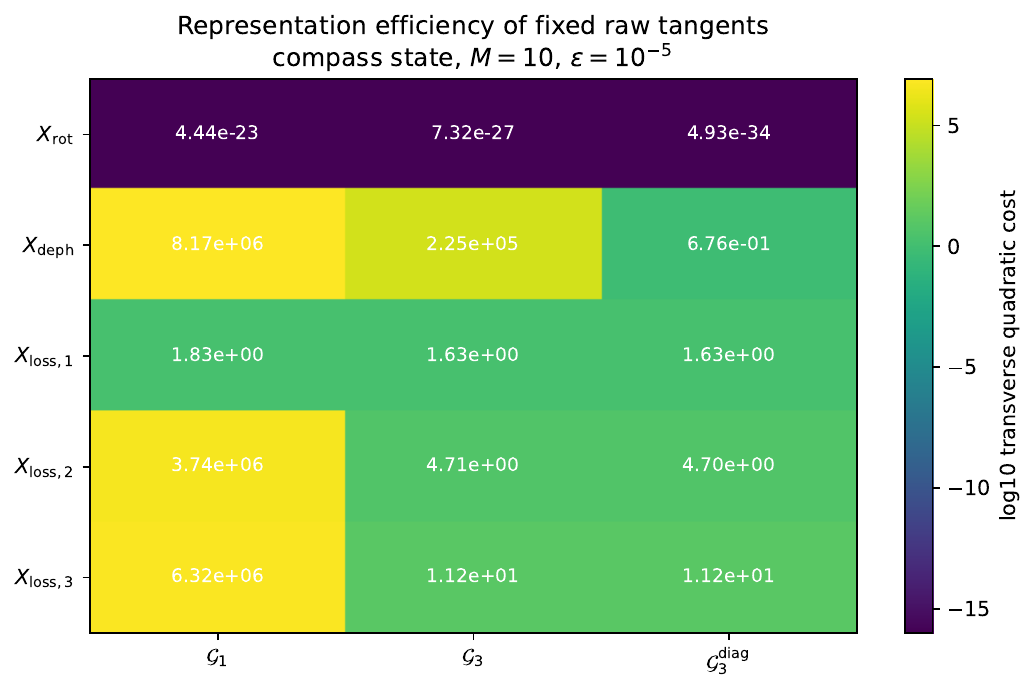}
\caption{
Representation efficiency of fixed raw tangents. The heatmap shows $\log_{10}\mathfrak g_{\rho,\varepsilon}^{\perp,\Gset,w}(X,X)$ for a compass-like state. Every frame already spans the trace-zero sector; higher-order and diagonal directions lower selected costs by providing direct, well-conditioned representations. The rotation has numerical-zero transverse cost in all cases.
}
\label{fig:gset-dependence}
\end{figure}

\subsection{Thermal weights and modular covariance}

We now orient the weighted frame with a truncated thermal reference
\[
\sigma_{\beta,M}
=
\frac{\ee^{-\beta\Nop_M}}{\Tr(\ee^{-\beta\Nop_M})}.
\]
The exact finite-dimensional modular relations are
\[
\sigma_{\beta,M}a_M^k\sigma_{\beta,M}^{-1}
=
\ee^{\beta k}a_M^k,
\qquad
\sigma_{\beta,M}(a_M^\dagger)^k\sigma_{\beta,M}^{-1}
=
\ee^{-\beta k}(a_M^\dagger)^k.
\]
This supplies a level-B modular justification: the operator frame has the correct modular covariance and the resistance weights are chosen as inverse rates. Writing $\gamma_k^-$ and $\gamma_k^+$ for loss and gain rates, we impose
\[
\frac{\gamma_k^+}{\gamma_k^-}=\ee^{-\beta k},
\qquad
w_k^\pm\propto\frac1{\gamma_k^\pm},
\]
and fix the common scale by
\[
w_k^-=1,
\qquad
w_k^+=\ee^{\beta k}.
\]
The frame covariance $V_j\mapsto c_jV_j$, $w_j\mapsto|c_j|^2w_j$ guarantees that this statement is independent of scalar jump-operator conventions.

Figure~\ref{fig:thermal-weights} uses the \emph{raw} Lindblad tangents
\[
X_{\mathrm{loss},k}=\Diss[L_k](\rho),
\qquad
X_{\mathrm{gain},k}=\Diss[L_k^\dagger](\rho),
\]
with no tangent normalization, and plots
\[
R_k(\beta)=
\frac{\mathfrak g_{\rho,\varepsilon,w_\beta}^{\perp}(X_{\mathrm{gain},k},X_{\mathrm{gain},k})}
{\mathfrak g_{\rho,\varepsilon,w_\beta}^{\perp}(X_{\mathrm{loss},k},X_{\mathrm{loss},k})}.
\]
At $\beta=0$ the ratios need not equal one because the gain and loss tangents themselves have different Hilbert--Schmidt amplitudes and state dependence. Over $0\leq\beta\leq4$, $R_1$ remains near one, $R_2$ grows from $2.98$ to $5.60$, and $R_3$ grows from $8.65$ to $108$.

This covariance does \emph{not} constitute a complete KMS transport geometry. The mobility remains the arithmetic mean $\mathcal K_\rho(Y)=(\rho Y+Y\rho)/2$, rather than a logarithmic or Kubo--Mori mobility tied to a KMS-symmetric Dirichlet form. The figure therefore demonstrates modularly oriented frame weights, not an entropy-gradient or full detailed-balance geometry.

\begin{figure}[t]
\centering
\includegraphics[width=0.7\linewidth]{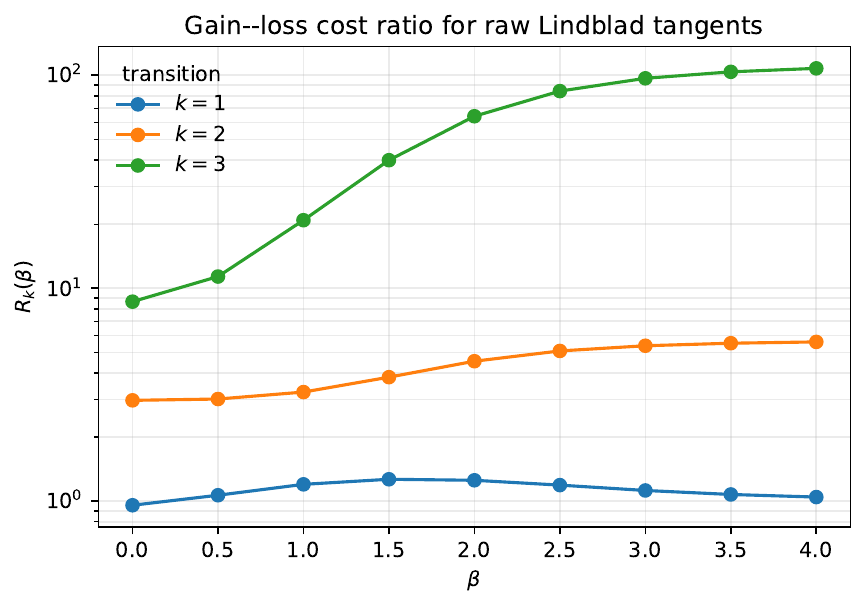}
\caption{
Gain--loss cost ratios for raw, unnormalized Lindblad tangents of a compass-like state, with $M=10$ and $\varepsilon=10^{-5}$. The resistance weights satisfy $w_k^-=1$ and $w_k^+=\exp(\beta k)$, corresponding to inverse rates with $\gamma_k^+/\gamma_k^-=\exp(-\beta k)$. The operator frame has exact modular covariance, but the arithmetic mobility means that this is not a complete KMS geometry.
}
\label{fig:thermal-weights}
\end{figure}

\subsection{Cotangent diagnostic on the diagonal Fock graph}

A full dual of the transverse seminorm requires care. Since every Hamiltonian tangent $Y=-\ii[H,\rho]$ has zero transverse seminorm, the formal quotient dual
\[
\sup_{X\neq0}
\frac{|\Tr(AX)|^2}{\mathfrak g_{\rho,\varepsilon}^{\perp,\Gset,w}(X,X)}
\]
is finite only if the observable annihilates all unitary tangents. By cyclicity of the trace, this condition is equivalent to
\[
[A,\rho]=0.
\]
For a generic compass state, diagonal Fock observables do not satisfy this condition, so Figure~\ref{fig:observable-dual} must not be interpreted as the full dual of the compass-state problem.

Instead, the figure is a controlled cotangent diagnostic on the diagonal Fock graph at the diagonal thermal reference $p_n\propto\ee^{-\beta_0n}$. The operators $L_k$ induce edges between levels separated by $k$, with conductance equal to the squared Fock transition amplitude times the arithmetic population mobility. For a diagonal observable $A=f(\Nop)$ we compute the graph cotangent energy
\[
\|A\|_{\mathrm{cot},\Gset_K}^2
=
\sum_{1\leq k\leq K}\sum_{n\geq k}
 c_{n-k,n}^{(k)}|f(n)-f(n-k)|^2.
\]
Because both the reference and the observables are diagonal, $[A,\rho_{\mathrm{diag}}]=0$ and the finite-dual obstruction is absent in this restricted problem.

We plot
\[
Q_A=
\frac{\|A\|_{\mathrm{cot},\Gset_3}}
{\|A\|_{\mathrm{cot},\Gset_1}}.
\]
The largest enhancements occur for energy-like and high-photon-number observables: $Q_N\simeq10.7$, $Q_{N^2}\simeq19.9$, and $Q_{\Pi_{n\geq6}}\simeq15.7$. This is an observable-side graph diagnostic of the chosen transitions, not a complete cotangent characterization of the noncommutative compass geometry.

\begin{figure}[t]
\centering
\includegraphics[width=0.7\linewidth]{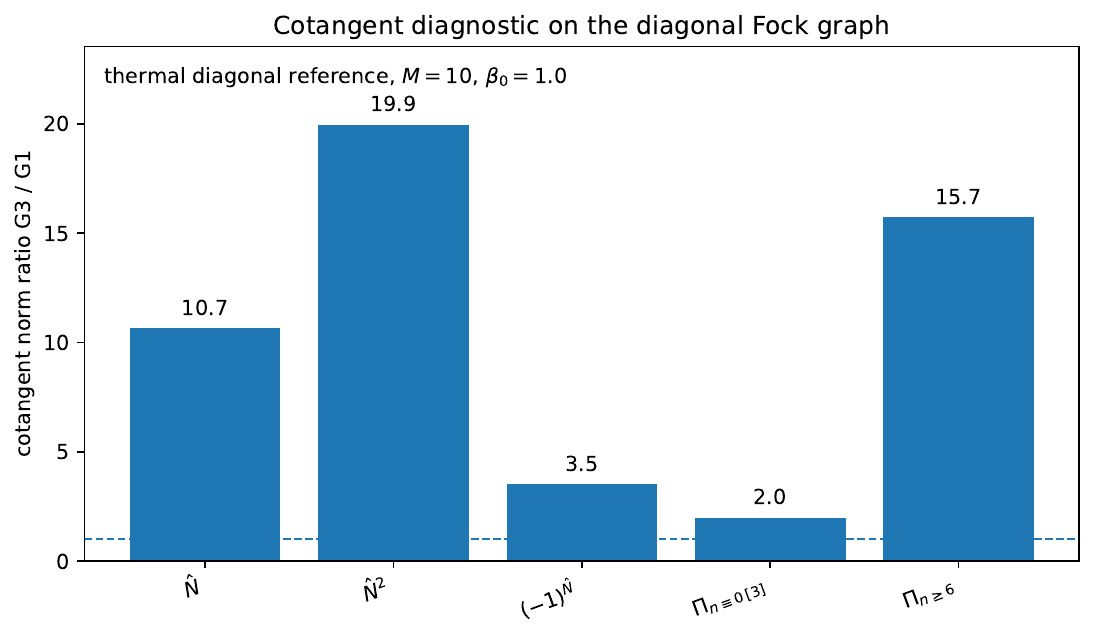}
\caption{
Cotangent diagnostic on the diagonal Fock graph, not the full transverse dual of the compass-state problem. The ratio compares graph cotangent energies for the one-photon and one-to-three-photon edge sets at a diagonal thermal reference with $M=10$ and $\beta_0=1$. Diagonal observables commute with this reference, satisfying the finiteness condition for the restricted transverse dual.
}
\label{fig:observable-dual}
\end{figure}

\subsection{Numerical conclusions}

The diagnostics establish four distinct points. The quotient removes the Hamiltonian phase rotation. The one-photon frame is already algebraically complete on the finite trace-zero sector, but higher-order and diagonal frame elements can reduce selected representation costs by many orders of magnitude. Modularly oriented inverse-rate weights polarize raw gain and loss tangents without defining a full KMS geometry. The observable-side calculation is valid as a diagonal graph cotangent diagnostic and is deliberately not identified with the full compass-state dual.

All conclusions are local and normalization-dependent. Image residuals, cutoff variation, regularization variation, quadratic homogeneity and weighted-frame covariance are reported in Appendix~\ref{app:robustness} and in the accompanying data files.

\section{Relation to existing quantum Wasserstein viewpoints}
\label{sec:relations}

The construction introduced above is inspired by optimal transport ideas, but it should not be confused with a general quantum Wasserstein distance. Its purpose is more limited and more directly operational: it is a local test for deciding whether an observed or modeled infinitesimal variation is a reversible motion or a dissipative change compatible with a chosen environment.

In this section we clarify the relation with three nearby viewpoints: noncommutative transport metrics associated with quantum Markov semigroups, quantum Wasserstein distances of order one, and measurement-aware distinguishability.

\subsection{Noncommutative transport and entropy gradient flows}

A major line of work in noncommutative optimal transport constructs transport metrics from quantum Markov semigroups satisfying detailed balance. In the approach of Carlen and Maas, and in related developments by Wirth, the metric is tied to a noncommutative continuity equation, a mobility operator and a Dirichlet-form structure \cite{carlen_maas_2017,wirth_2018}. Under suitable assumptions, the quantum Markov semigroup becomes the gradient flow of the relative entropy.

Our construction shares some of these ingredients. We also use a noncommutative continuity equation, a mobility and a quadratic flux cost. However, the starting point is different. We do not begin with a fixed quantum Markov semigroup and ask for the metric that makes it an entropy gradient flow. Instead, we begin with a finite list of physically selected directions
\[
\Gset=\{V_j\}_j,
\]
such as one-photon and multiphoton transitions, and ask how efficiently selected tangent vectors are represented by them.

The difference can be summarized as follows: semigroup-based transport builds a metric adapted to a given dynamics, whereas the present construction assigns a local representation cost to a weighted operator frame and then quotients Hamiltonian tangents. The frame is not asserted to determine a unique GKSL generator.

This distinction is important. The present framework is not designed to prove an entropy gradient-flow structure. It is designed to compare reversible and dissipative components of variations of non-Gaussian oscillator states.

\subsection{Quantum Wasserstein distances of order one}

Another relevant viewpoint is provided by quantum Wasserstein distances of order one. In the construction of De Palma, Marvian, Trevisan and Lloyd, the distance is connected with a notion of locality and with a dual characterization in terms of Lipschitz observables \cite{depalma_marvian_trevisan_lloyd_2021}. Related work by De Palma and Trevisan develops quantum optimal-transport structures using quantum channels \cite{depalma_trevisan_2021}.

This viewpoint is close to one aspect of our construction: the set \(\Gset\) defines a notion of neighborhood or elementary transition. In the diagonal Fock basis, the directions \(L_k\) connect occupation numbers separated by \(k\), and therefore define a graph-like structure on the truncated number basis. In this sense, changing \(\Gset\) changes the notion of locality.

However, our object is not a global $W_1$ distance. It is a local quadratic cost on tangent vectors, regularized at a state $\rho$, and then quotiented by unitary directions. A full transverse dual is finite only on observables commuting with $\rho$; the numerical observable-side example is therefore restricted to a diagonal thermal graph. The construction is closer to a local representation and cotangent diagnostic than to a canonical global distance.

The relation with \(W_1\)-type ideas is nevertheless useful. It clarifies why changing the available transitions has two complementary effects. On the primal side, adding directions may reduce the cost of certain state variations. On the dual side, it changes which observables are considered regular or sensitive. This is why the numerical diagnostics include transverse costs together with a deliberately restricted graph cotangent norm.

\subsection{Why the transverse quotient is not standard}

The main feature that distinguishes the present construction is the transverse quotient
\[
X
\sim
X+\ii[H,\rho].
\]
This quotient removes the tangent space to the unitary orbit before assigning a dissipative cost.

Standard quantum Wasserstein distances do not usually remove Hamiltonian directions in this way. They may measure the distance between states, or control expectation-value differences for Lipschitz observables, but they do not directly answer the following local question:
\[
\text{which part of }X
\text{ is dissipative after subtracting the best unitary motion?}
\]

This question is natural for open oscillator systems. A phase rotation of a cat or compass state may strongly change its Wigner representation, but it is reversible. Dephasing, loss and gain are different: they generally change the spectrum or destroy coherences. The transverse quotient is introduced precisely to make this distinction explicit.

\subsection{Relation with measurement-aware distinguishability}

The motivation is also related to measurement-aware distinguishability. In a previous preprint, we studied how different levels of description---state-space transport, projected quantities and detector statistics---lead to different notions of useful distinguishability in lossy quantum sensing \cite{our_measurement_aware_2026}.

The present work follows the same methodological principle: a distance or cost is useful only relative to a physical context. The difference is that the context is now encoded by environmental directions rather than by a measurement chain. In this sense, the present paper moves from measurement-aware distinguishability to environment-aware transverse transport.

The two perspectives are complementary. Measurement-aware distinguishability asks which differences remain visible after measurement. Environment-aware transverse transport asks how the non-Hamiltonian part of a variation is represented by a specified environmental frame. The diagonal graph cotangent diagnostic provides a limited observable-side bridge between the two viewpoints.

\subsection{Summary of the positioning}

The present framework should therefore be read as a local diagnostic rather than as a competing general theory of quantum transport. Compared with semigroup-based noncommutative transport, it is more directly tied to selected physical directions. Compared with \(W_1\)-type quantum Wasserstein distances, it is local, quadratic and state-dependent. Compared with measurement-aware distinguishability, it focuses on the environmental origin of variations rather than on the measurement chain.

Its intended use is consequently narrow but concrete: it quantifies the dissipative representation cost of non-Gaussian oscillator variations relative to a chosen environment.

\section{Conclusion}
\label{sec:conclusion}

We introduced an environment-aware local transport cost for infinitesimal variations of truncated non-Gaussian oscillator states. Its square root after quotienting by Hamiltonian tangents is a transverse seminorm, not a global distance. The current quotient uses all Hermitian Hamiltonians; a restricted reversible-control space $\mathfrak h_{\mathrm{rev}}$ is a natural future refinement.

A central finite-dimensional fact is that $\Gset_1=\{a_M,a_M^\dagger\}$ already generates the full trace-zero sector through the divergence. The numerical role of multiphoton and diagonal operators is therefore to improve representation efficiency and direct physical adaptation. This distinction removes the misleading claim that higher-order loss tangents first become algebraically accessible when their jump operators are inserted.

The model datum is the weighted operator frame $(\Gset,w)$, not a unique GKSL decomposition. The covariance $V_j\mapsto c_jV_j$, $w_j\mapsto|c_j|^2w_j$ makes the cost invariant under scalar jump-operator reparameterization. Thermal inverse-rate weights are supported by the exact modular covariance of $a_M^k$ and $(a_M^\dagger)^k$, while the arithmetic mobility prevents interpreting the construction as a complete KMS geometry.

The numerical figures use explicitly stated raw tangents, including $X_{\mathrm{deph}}=-[N,[N,\rho]]=2\Diss[N](\rho)$. The experimental tangent may be estimated as $(\rho_{\mathrm{out}}-\rho_{\mathrm{in}})/\Delta t$, with the resulting quadratic time normalization reported. Image residuals and optimization residuals are small, the principal hierarchy is stable under the tested cutoffs and regularizations, and frame covariance and quadratic homogeneity hold at approximately $10^{-12}$ relative accuracy.

Finally, the observable-side figure is only a cotangent diagnostic on a diagonal Fock graph. A finite full transverse dual requires $[A,\rho]=0$; this condition generally fails for diagonal observables at a compass state. Future work may replace the arithmetic mobility by a logarithmic/Kubo--Mori mean, restrict the reversible Hamiltonian quotient, and combine environmental representation costs with measurement-aware distinguishability.

\appendix

\section{Elementary properties}
\label{app:elementary-properties}

This appendix collects a few elementary properties of the local and transverse costs. The purpose is not to develop a global metric theory, but to make explicit the finite-dimensional facts used in the main text.

Throughout the appendix, the Fock cutoff \(M\), the state \(\rho\), the regularization parameter \(\varepsilon>0\), the direction family
\[
\Gset=\{V_j\}_{j=1}^r,
\]
and the positive weights \(w_j>0\) are fixed. We write
\[
\rho_\varepsilon=(1-\varepsilon)\rho+\varepsilon\frac{\Id}{d},
\qquad
d=M+1.
\]

\subsection{Positivity of the flux cost}

The regularized mobility is
\[
\mathcal K_{\rho_\varepsilon}(Y)
=
\frac12
\left(
\rho_\varepsilon Y+Y\rho_\varepsilon
\right).
\]
Since \(\rho_\varepsilon\) is strictly positive, the map \(\mathcal K_{\rho_\varepsilon}\) is positive and invertible on the finite-dimensional operator space \(\Aalg_M\).

\begin{proposition}[Positivity]
For every flux \(J=(J_1,\ldots,J_r)\),
\[
\mathcal C_{\rho_\varepsilon,w}(J)\geq 0.
\]
Consequently, for every tangent vector \(X\),
\[
\mathfrak g_{\rho,\varepsilon}^{\Gset,w}(X,X)\geq 0,
\]
with the convention that the value may be \(+\infty\).
\end{proposition}

\begin{proof}
For each \(j\), positivity of \(\mathcal K_{\rho_\varepsilon}^{-1}\) gives
\[
\operatorname{Re}
\left\langle
J_j,
\mathcal K_{\rho_\varepsilon}^{-1}(J_j)
\right\rangle_{\mathrm{HS}}
\geq 0.
\]
Since \(w_j>0\), each term in the sum defining
\[
\mathcal C_{\rho_\varepsilon,w}(J)
=
\sum_{j=1}^r
w_j
\operatorname{Re}
\left\langle
J_j,
\mathcal K_{\rho_\varepsilon}^{-1}(J_j)
\right\rangle_{\mathrm{HS}}
\]
is non-negative. The local cost is the infimum of non-negative quantities over the affine constraint
\[
X+\operatorname{div}_{\Gset}(J)=0.
\]
This proves the result.
\end{proof}

\subsection{Characterization of zero raw cost}

The raw local cost can vanish only when the tangent vector is represented by a zero-cost flux. With strictly positive weights and strictly positive mobility, the only zero-cost flux is the zero flux.

\begin{proposition}[Zero raw cost]
If
\[
\mathfrak g_{\rho,\varepsilon}^{\Gset,w}(X,X)=0
\]
and if the infimum is attained, then \(X=0\).
\end{proposition}

\begin{proof}
If the infimum is attained by a flux \(J\), then
\[
\mathcal C_{\rho_\varepsilon,w}(J)=0.
\]
By positivity of each term and by strict positivity of the weights, this implies
\[
J_j=0
\qquad
\text{for all }j.
\]
The continuity constraint then gives
\[
X+\operatorname{div}_{\Gset}(J)=X=0.
\]
\end{proof}

\begin{remark}
In finite dimension, the infimum is attained whenever the constraint set is nonempty, after restricting the flux variables to the finite-dimensional range relevant to the divergence map. In numerical implementations this is handled by the Moore--Penrose inverse of the corresponding quadratic problem.
\end{remark}

\subsection{Monotonicity under enlargement of the direction set}

The local cost depends on the admissible directions. Adding new directions can only enlarge the set of feasible fluxes, provided the old directions and weights are kept unchanged.

Let
\[
\Gset=\{V_1,\ldots,V_r\}
\]
and let
\[
\Gset'=\{V_1,\ldots,V_r,V_{r+1},\ldots,V_{r+s}\}
\]
be an enlarged family. Assume that the weights of the first \(r\) directions are the same in both models.

\begin{proposition}[Monotonicity under additional directions]
For every tangent vector \(X\),
\[
\mathfrak g_{\rho,\varepsilon}^{\Gset',w'}(X,X)
\leq
\mathfrak g_{\rho,\varepsilon}^{\Gset,w}(X,X),
\]
where \(w'_j=w_j\) for \(1\leq j\leq r\).
\end{proposition}

\begin{proof}
Any flux
\[
J=(J_1,\ldots,J_r)
\]
for the direction family \(\Gset\) can be embedded into a flux for \(\Gset'\) by setting
\[
J'=(J_1,\ldots,J_r,0,\ldots,0).
\]
The divergence of \(J'\) with respect to \(\Gset'\) is the same as the divergence of \(J\) with respect to \(\Gset\), and the cost is unchanged because the added flux components vanish. Therefore the feasible set for \(\Gset'\) contains an isometric copy of the feasible set for \(\Gset\). Taking the infimum gives the result.
\end{proof}

\begin{remark}
This monotonicity concerns the raw local cost. In practice, adding directions may also change the modeling interpretation of the environment. The inequality simply says that, when all previous directions and weights are retained, additional admissible flux channels cannot increase the minimal cost.
\end{remark}

\subsection{Vanishing of the transverse cost on unitary directions}

The transverse cost is defined by
\[
\mathfrak g_{\rho,\varepsilon}^{\perp,\Gset,w}(X,X)
=
\inf_{H=H^\dagger}
\mathfrak g_{\rho,\varepsilon}^{\Gset,w}
\left(
X+\ii[H,\rho],
X+\ii[H,\rho]
\right).
\]

\begin{proposition}[Unitary directions have zero transverse cost]
If
\[
X=-\ii[H_0,\rho]
\]
for some Hermitian operator \(H_0\), then
\[
\mathfrak g_{\rho,\varepsilon}^{\perp,\Gset,w}(X,X)=0.
\]
\end{proposition}

\begin{proof}
Choosing \(H=H_0\) gives
\[
X+\ii[H,\rho]
=
-\ii[H_0,\rho]+\ii[H_0,\rho]
=
0.
\]
The raw local cost of the zero tangent vector is zero, because it is represented by the zero flux. Hence the infimum defining the transverse cost is zero.
\end{proof}

\subsection{Quadratic quotient invariance}

The previous proposition implies that the transverse cost is not positive definite on the full tangent space. It is positive only after quotienting by unitary directions.

Let
\[
T_\rho^{\mathrm{unit}}
=
\{-\ii[H,\rho]:H=H^\dagger\}.
\]
If two tangent vectors differ by an element of \(T_\rho^{\mathrm{unit}}\), then they have the same transverse cost.

\begin{proposition}[Invariance under unitary tangent shifts]
Let
\[
Y=-\ii[H_0,\rho]\in T_\rho^{\mathrm{unit}}.
\]
Then
\[
\mathfrak g_{\rho,\varepsilon}^{\perp,\Gset,w}(X+Y,X+Y)
=
\mathfrak g_{\rho,\varepsilon}^{\perp,\Gset,w}(X,X).
\]
\end{proposition}

\begin{proof}
By definition,
\begin{align*}
\mathfrak g_{\rho,\varepsilon}^{\perp,\Gset,w}(X+Y,X+Y)
&=
\inf_{H=H^\dagger}
\mathfrak g_{\rho,\varepsilon}^{\Gset,w}
\bigl(
X-\ii[H_0,\rho]+\ii[H,\rho], \\
&\hspace{4.3cm}
X-\ii[H_0,\rho]+\ii[H,\rho]
\bigr).
\end{align*}
Set
\[
\widetilde H=H-H_0.
\]
As \(H\) runs over all Hermitian operators, so does \(\widetilde H\). Hence the above infimum is equal to
\[
\inf_{\widetilde H=\widetilde H^\dagger}
\mathfrak g_{\rho,\varepsilon}^{\Gset,w}
\left(
X+\ii[\widetilde H,\rho],
X+\ii[\widetilde H,\rho]
\right),
\]
which is precisely
\[
\mathfrak g_{\rho,\varepsilon}^{\perp,\Gset,w}(X,X).
\]
\end{proof}

Thus the transverse cost is naturally defined on the quotient space
\[
T_\rho\mathcal D_M/T_\rho^{\mathrm{unit}}.
\]
This is the local mathematical expression of the physical rule used in the main text: reversible Hamiltonian variations are not counted as dissipative degradation.

\section{Numerical implementation}
\label{app:numerical-implementation}

This appendix describes the finite-dimensional numerical implementation of the local and transverse costs. All computations are performed after choosing a Fock cutoff \(M\), so that the Hilbert space dimension is
\[
d=M+1.
\]
Matrices in \(\Aalg_M\) are represented as \(d\times d\) complex matrices and then vectorized.

\subsection{Vectorization convention}

We use the column-stacking vectorization
\[
\operatorname{vec}(A)_{m+nd}=A_{mn},
\]
so that
\[
\operatorname{vec}(AXB)
=
(B^{\mathsf T}\otimes A)\operatorname{vec}(X).
\]
With this convention, left and right multiplication by a matrix are represented by
\[
L_A:\operatorname{vec}(X)\mapsto \operatorname{vec}(AX)
=
(\Id\otimes A)\operatorname{vec}(X),
\]
and
\[
R_A:\operatorname{vec}(X)\mapsto \operatorname{vec}(XA)
=
(A^{\mathsf T}\otimes \Id)\operatorname{vec}(X).
\]

\subsection{Matrix representation of the mobility}

The symmetric mobility is
\[
\mathcal K_{\rho_\varepsilon}(Y)
=
\frac12(\rho_\varepsilon Y+Y\rho_\varepsilon).
\]
In vectorized form, it is represented by the positive matrix
\[
K_{\rho_\varepsilon}
=
\frac12
\left(
\Id\otimes \rho_\varepsilon
+
\rho_\varepsilon^{\mathsf T}\otimes \Id
\right).
\]
Because
\[
\rho_\varepsilon>0,
\]
the matrix \(K_{\rho_\varepsilon}\) is invertible.

For a flux component \(J_j\), the quadratic contribution to the cost is
\[
w_j
\operatorname{Re}
\left\langle
J_j,\mathcal K_{\rho_\varepsilon}^{-1}(J_j)
\right\rangle_{\mathrm{HS}}.
\]
In vectorized form, this becomes
\[
w_j
\,\operatorname{Re}
\left[
\operatorname{vec}(J_j)^\dagger
K_{\rho_\varepsilon}^{-1}
\operatorname{vec}(J_j)
\right].
\]

\subsection{Matrix representation of the divergence}

For a direction \(V_j\), the divergence contribution is
\[
J_jV_j^\dagger
-
V_j^\dagger J_j.
\]
Using the vectorization convention above,
\[
\operatorname{vec}(J_jV_j^\dagger)
=
((V_j^\dagger)^{\mathsf T}\otimes \Id)\operatorname{vec}(J_j),
\]
and
\[
\operatorname{vec}(V_j^\dagger J_j)
=
(\Id\otimes V_j^\dagger)\operatorname{vec}(J_j).
\]
Thus the contribution of \(J_j\) to the divergence is represented by
\[
D_j
=
(V_j^\ast\otimes \Id)
-
(\Id\otimes V_j^\dagger),
\]
where \(V_j^\ast=(V_j^\dagger)^{\mathsf T}\).

The full divergence matrix is the block matrix
\[
D_{\Gset}
=
\begin{bmatrix}
D_1 & D_2 & \cdots & D_r
\end{bmatrix}.
\]
If
\[
j=
\begin{bmatrix}
\operatorname{vec}(J_1)\\
\vdots\\
\operatorname{vec}(J_r)
\end{bmatrix},
\]
then
\[
\operatorname{vec}(\operatorname{div}_{\Gset}(J))
=
D_{\Gset}j.
\]

The continuity equation
\[
X+\operatorname{div}_{\Gset}(J)=0
\]
becomes the linear constraint
\[
D_{\Gset}j=-x,
\qquad
x=\operatorname{vec}(X).
\]

\subsection{Quadratic minimization for the raw cost}

Let
\[
G_{\rho_\varepsilon,w}
=
\operatorname{diag}
\left(
w_1K_{\rho_\varepsilon}^{-1},
\ldots,
w_rK_{\rho_\varepsilon}^{-1}
\right).
\]
The raw cost is the constrained quadratic problem
\[
\mathfrak g_{\rho,\varepsilon}^{\Gset,w}(X,X)
=
\inf_{j:\,D_{\Gset}j=-x}
j^\dagger
G_{\rho_\varepsilon,w}
j.
\]

Equivalently, after introducing
\[
B_{\rho_\varepsilon,w}
=
D_{\Gset}
G_{\rho_\varepsilon,w}^{-1}
D_{\Gset}^\dagger,
\]
one obtains the dual expression
\[
\mathfrak g_{\rho,\varepsilon}^{\Gset,w}(X,X)
=
x^\dagger
B_{\rho_\varepsilon,w}^{+}
x,
\]
provided \(x\) belongs to the range of \(D_{\Gset}\). Here \(B^+\) denotes the Moore--Penrose pseudoinverse. In numerical computations, small singular values are removed according to a fixed tolerance.

This formula is convenient because it reduces the flux minimization to a pseudoinverse on the tangent-vector space.

\subsection{Hermiticity and real representation}

The physical tangent vectors are Hermitian and traceless. The flux variables may be represented either as complex matrices with the final cost restricted to real tangent directions, or by separating real and imaginary parts. In the computations reported in the paper, one may use the real representation
\[
\operatorname{vec}_{\mathbb R}(A)
=
\begin{bmatrix}
\operatorname{Re}\operatorname{vec}(A)\\
\operatorname{Im}\operatorname{vec}(A)
\end{bmatrix}.
\]
The complex linear maps are then replaced by their real block representations. This ensures that the resulting quadratic forms are real symmetric matrices.

The trace constraint
\[
\Tr(X)=0
\]
is imposed by projecting out the identity direction. Equivalently, the computations can be performed in the full matrix space and then restricted to traceless tangent vectors.

\subsection{Numerical transverse quotient}

The transverse cost is
\[
\mathfrak g_{\rho,\varepsilon}^{\perp,\Gset,w}(X,X)
=
\inf_{H=H^\dagger}
\mathfrak g_{\rho,\varepsilon}^{\Gset,w}
\left(
X+\ii[H,\rho],
X+\ii[H,\rho]
\right).
\]

Let
\[
A_\rho(H)=\ii[H,\rho].
\]
After choosing a real basis \(\{F_\alpha\}_\alpha\) of Hermitian matrices, write
\[
H=\sum_\alpha h_\alpha F_\alpha.
\]
The map
\[
h\mapsto \operatorname{vec}_{\mathbb R}(A_\rho(H))
\]
is represented by a real matrix \(U_\rho\).

Let \(B^+\) denote the raw-cost matrix acting on tangent vectors. The transverse cost is then the finite-dimensional least-squares problem
\[
\inf_h
(x+U_\rho h)^{\mathsf T}
B^+
(x+U_\rho h).
\]
When the relevant matrices are restricted to the finite-dimensional range of the raw cost, the minimizer satisfies the normal equation
\[
U_\rho^{\mathsf T}B^+U_\rho h
=
-
U_\rho^{\mathsf T}B^+x.
\]
Using the Moore--Penrose pseudoinverse if necessary, one obtains
\[
h_\ast
=
-
\left(
U_\rho^{\mathsf T}B^+U_\rho
\right)^+
U_\rho^{\mathsf T}B^+x.
\]
The transverse cost is then
\[
(x+U_\rho h_\ast)^{\mathsf T}
B^+
(x+U_\rho h_\ast).
\]

This implementation makes explicit that the transverse quotient is a projection with respect to the raw local cost, not with respect to the Hilbert--Schmidt norm.

\subsection{Regularization and tolerances}

The parameter \(\varepsilon\) controls the conditioning of the mobility. For nearly pure states, too small a value of \(\varepsilon\) can produce large condition numbers. The numerical experiments should therefore report the value of \(\varepsilon\) and check that qualitative conclusions are stable over a reasonable range.

Similarly, all pseudoinverses depend on a singular-value cutoff. The same tolerance should be used consistently across comparable experiments. A typical implementation keeps singular values larger than
\[
\tau\,s_{\max},
\]
where \(s_{\max}\) is the largest singular value and \(\tau\) is a prescribed relative tolerance.

The numerical values reported in the main text should therefore be interpreted as regularized local costs. The robust information is contained in the hierarchy of costs and in their dependence on \(\Gset\), \(w\), and \(\rho\), rather than in any single absolute number.

\section{Numerical scope and stability}
\label{app:robustness}

The computations use the column-vectorization and pseudoinverse formulas described above. Every reported tangent is Hermitian and trace zero. For each calculation we record the relative image residual
\[
r_{\mathrm{im}}(X)=
\frac{\|(I-P_{\operatorname{Im}D})\operatorname{vec}(X)\|_2}
{\|\operatorname{vec}(X)\|_2},
\]
and, after the Hamiltonian minimization, the relative residual of the normal equations. The data files accompanying Figures~2--4 contain these quantities for every plotted point.

\begin{table}[t]
\centering
\caption{Selected numerical controls. Costs refer to the raw tangents used in the figures.}
\label{tab:numerical-controls}
\small
\begin{tabular}{p{0.31\linewidth}p{0.34\linewidth}p{0.24\linewidth}}
\toprule
Control & Range tested & Result \\
\midrule
$\operatorname{rank}D_{\Gset_1}$ & $M=4,6,8,10,12$ & exactly $(M+1)^2-1$ \\
Image residual, cutoff tests & $M=8,10,12$, $\varepsilon=10^{-5}$ & at most $3.2\times10^{-11}$ \\
$g_\perp(X_{\mathrm{loss},2};\Gset_3)$ & $M=8,10,12$ & $4.688,\ 4.710,\ 4.701$ \\
$g_\perp(X_{\mathrm{deph}};\Gset_3^{\mathrm{diag}})$ & $M=8,10,12$ & $0.674,\ 0.676,\ 0.683$ \\
Frame covariance error & complex rescalings of the six $\Gset_3$ elements & $2.1\times10^{-12}$ \\
Quadratic homogeneity error & $\lambda=2.5$ & $1.1\times10^{-12}$ \\
\bottomrule
\end{tabular}
\end{table}

The regularization test at $M=10$ uses $\varepsilon=10^{-4},10^{-5},10^{-6}$. The directly adapted $\Gset_3$ cost of $X_{\mathrm{loss},2}$ remains $4.7088$, $4.7098$, and $4.7099$, while the $\Gset_3^{\mathrm{diag}}$ dephasing cost remains $0.626$, $0.676$, and $0.681$. By contrast, the indirect $\Gset_1$ representation of $X_{\mathrm{loss},2}$ scales from $3.74\times10^5$ to $3.74\times10^7$ as $\varepsilon$ decreases. This strong dependence is expected: the inverse arithmetic mobility penalizes components supported on small eigenvalues of the nearly pure reference state. The robust conclusion is the efficiency hierarchy, not a universal absolute number.

Changing the cutoff from $M=8$ to $M=12$ leaves the directly adapted costs nearly unchanged and preserves the many-orders-of-magnitude contrast between $\Gset_1$ and $\Gset_3$. Boundary populations of the chosen compass state are small over this range. The algebraic image rank is exactly the trace-zero dimension in every tested truncation, numerically confirming the proposition for $\Gset_1$.

The pseudoinverse tolerance is $10^{-10}$ relative to the largest eigenvalue. At the main parameters, image and normal-equation residuals are well below the scale of the plotted effects. All raw values, residuals, tangent Hilbert--Schmidt norms, cutoffs, regularizations and frame conventions are provided in comma-separated data files and regenerated by the supplied scripts.

\printbibliography

\end{document}